\documentclass[letterpaper, 10 pt, conference]{ieeeconf} %

\IEEEoverridecommandlockouts %
\usepackage{subcaption}
\usepackage{booktabs}
\usepackage{url}
\usepackage[hidelinks]{hyperref}
\usepackage{cite}
\usepackage{amsmath, amssymb, amsfonts, stmaryrd}
\usepackage{algorithmic}
\usepackage{graphicx}
\usepackage{textcomp}
\usepackage{xcolor}
\def\BibTeX{{\rm B\kern-.05em{\sc i\kern-.025em b}\kern-.08em T\kern-.1667em\lower.7ex\hbox{E}\kern-.125emX}}
\newtheorem{thm}{Theorem}%
\newtheorem{proposition}{Proposition}%
\newtheorem{defn}{Definition}%
\newtheorem{rem}{Remark}%
\newtheorem{assumption}{Assumption}%

\usepackage[ruled, vlined, linesnumbered]{algorithm2e}
\usepackage{makecell}

\usepackage{enumerate}

\newcommand{\set}[1]{\left\{#1\right\}}

\newcommand{\ra}{\rightarrow}
\newcommand{\Real}{\mathbb{R}}

\renewcommand{\subset}{\subseteq}

\newcommand{\LTL}{\mathrm{LTL}_{\setminus \bigcirc}}
\newcommand{\Sm}{\mathcal{S}_m}
\newcommand{\Gm}{\mathcal{G}_m}
\newcommand{\Am}{\mathcal{A}_m}
\newcommand{\Tm}{\mathcal{T}_m}
\newcommand{\lbl}[1]{\llbracket #1 \rrbracket}
\newcommand{\Win}{\operatorname{Win}}
\newcommand{\len}{\operatorname{len}}

\newcommand{\C}{\mathcal{C}}
\newcommand{\Cm}{\mathcal{C}^{m}}
\newcommand{\D}{\mathcal{D}}
\newcommand{\R}{\mathbb{R}}
\newcommand{\X}{\mathcal{X}}
\newcommand{\Wo}{\mathcal{W}^1}
\newcommand{\Wm}{\mathcal{W}^1_m}
\newcommand{\Xm}{\mathcal{X}_m}
\newcommand{\xs}{x_{m}^{\star}}

\newcommand{\Sl}{\Phi^{\mathrm{sol}}}
\newcommand{\xe}{x_{\mathrm{eq}}}
\newcommand{\hwin}{\widehat{\operatorname{Win}}}
\newcommand{\hwint}{\widehat{\Win_{\sigma}}(\Tm)}

\newcommand{\Ubox}{\mathcal{U}_{\mathrm{box}}}

\newcommand{\yt}[1]{{\color{black} #1}} 
\newcommand{\rz}[1]{{\color{black} #1}}

\title{\LARGE \bf Patching Control Lyapunov Barrier Functions for Temporal Logic
Specifications with Bounded Controls }

\author{Ruikun Zhou*, Yating Yuan*, Haocheng Chang, Yinan Li, and Yiming Meng%
\thanks{Ruikun Zhou is with the Department of Aeronautics and Astronautics, Massachusetts Institute of Technology, Cambridge, MA 02139, USA. Email: \texttt{ruikun@mit.edu}.}
\thanks{Yating Yuan, and Haocheng Chang are with the Department of Applied Mathematics, University of Waterloo, Waterloo, Ontario N2L 3G1, Canada. Emails: \texttt{\{yating.yuan, haocheng.chang\}@uwaterloo.ca} }%
\thanks{Yinan Li is with MacLean Engineering \& Marketing Co. Limited, Collingwood, ON, Canada. Email: \texttt{yli@macleanengineering.com}.}
\thanks{Yiming Meng is with the AI Thrust, Information Hub, Hong Kong University of Science and Technology (Guangzhou), China. Email: \texttt{yimingmeng@hkust-gz.edu.cn}.}
}

\begin{document}
\maketitle
\thispagestyle{empty}
\pagestyle{empty}

\begin{abstract}
We propose an abstraction-free framework for controller synthesis for continuous-time dynamical systems subject to Linear Temporal Logic (LTL) specifications and bounded control inputs. The proposed method combines the sequential decomposition of LTL tasks with the use of formally certified Control Lyapunov-Barrier Functions (CLBFs). By formulating local specifications as a sequence of safe-stabilization problems, we systematically approximate and patch the winning sets of the decomposed subtasks. The satisfaction of these local constraints is guaranteed by the offline-computed level sets of the CLBFs. As a result, our framework yields formally verified switching feedback controllers that enable efficient online planning and dynamic re-planning. This ensures robust continuous specification satisfaction in the presence of state perturbations, avoiding the explicit state-space abstractions commonly required in the literature. The approach is validated through numerical simulations and a hardware demonstration on a Crazyflie quadrotor.
\end{abstract}

\section{Introduction}
Linear Temporal Logic (LTL) serves as a rigorous formal language for specifying temporally structured mission requirements in robotics and autonomous systems, such as sequencing, repeated visitation, conditional task execution, and safety constraints over long horizons~\cite{clarke1999model}. It is therefore well-suited to motion planning and control problems that require the satisfaction of complex task specifications over time, and has been widely used in a range of robotic applications, such as autonomous driving, surveillance robots, and firefighting robots~\cite{luckcuck2019formal}. One of the central challenges in realizing LTL specifications in robotic systems is to design controllers for long-horizon tasks while accounting for dynamics and safety constraints \cite{kressgazit2018synthesis}.

Abstraction-based methods constitute a classical line of work for LTL control, where continuous systems are represented by finite symbolic models and controllers are synthesized to satisfy LTL specifications~\cite{kloetzer2008fully}. Building on this abstraction-based paradigm, reactive mission and motion planning frameworks have been developed to provide formal guarantees for high-level robotic tasks under admissible environment assumptions~\cite{kressgazit2009temporal}. A key step in such abstraction-based synthesis frameworks is the construction of finite abstractions of the continuous dynamics, with robustness margins used to account for mismatches between the control system and its finite abstractions~\cite{liu2016finite}. This line of work was further extended to robustly complete control synthesis for nonlinear systems through an interval branch-and-bound scheme with adaptive refinement~\cite{li2018rocs}. However, finite-abstraction-based synthesis can be computationally demanding due to the burden of abstraction construction and state-space explosion.

To address these limitations, another line of work pursues abstraction-free control synthesis under LTL specifications using control barrier functions (CBFs). A typical methodology is to decompose the LTL specification into sequential subtasks and then use CBF-based quadratic programs (CBF-QPs) to synthesize controllers~\cite{niu2020control}, but such methods still have technical restrictions. Notably, CBF-QP-based methods typically require online optimization to compute control inputs, and conflicts among barrier-function constraints may lead to infeasible QPs. To address this issue, methods for the composition of barrier functions and a prioritization-based control method were proposed to guarantee the feasibility of the controller~\cite{srinivasan2020barrier}. Nevertheless, these approaches still typically rely on online optimization to compute the control input. Alternatively, hybrid barrier certificates establish satisfaction of syntactically co-safe LTL specifications through a hybrid-system eventuality argument~\cite{bisoffi2018hybrid}. This approach avoids the standard CBF-QP formulation, but it does not explicitly incorporate bounded-control constraints.

On the other hand, recent progress on the compatibility of CBFs and CLFs shows that one can directly synthesize a single control Lyapunov barrier function (CLBF) for safe stabilization tasks \cite{dai2024verification, liu2025computing, liu2025formal}. More importantly, the level sets of CLBFs depict the invariant set where both safety and convergence are guaranteed. Therefore, in the scenarios where CLBFs are accessible and the LTL specification can be decomposed into a set of reach-avoid tasks, those level sets can be leveraged to capture a set of several feasible trajectories for the specifications, potentially without solving QP problems. 

In this work, by leveraging the formally certified CLBFs proposed in \cite{liu2025formal}, we propose an abstraction-free framework that synthesizes controllers and approximates the winning set for the given LTL specifications by approximating and patching the winning sets of decomposed subtasks. The main contributions are threefold, as follows.
\begin{itemize}
    \item To the authors' best knowledge, this is the first work connecting the approximation of the winning set of LTL specifications and the level sets of CLBFs under bounded control.
    \item By formulating the local specification tasks as a sequence of safe stabilization problems, the proposed method yields formally verified switching feedback controllers under bounded control constraints. With the offline-computed CLBFs, this formulation enables efficient online planning under bounded control constraints and allows for dynamic, on-the-fly re-planning that guarantees specification satisfaction in the presence of state perturbations while tracking reference trajectories in noisy environments or with perturbed dynamics.
    \item A set of case studies, including numerical simulations and 
    a real-world experiment in which a Crazyflie quadrotor tracks a trajectory generated by the proposed algorithm, is provided to demonstrate the effectiveness of the proposed method. 
\end{itemize}
    
\section{Preliminaries and Problem Formulation}

Throughout this paper, we consider a control-affine system of the form:
\begin{equation}
    \label{eq:sys}
    \dot{x}= f(x) + g(x)u,
\end{equation}
where $f:\X \subset \mathbb{R}^{n}\to \mathbb{R}^{n}$ and $g:\X \subset \mathbb{R}
^{n}\to \mathbb{R}^{n \times p}$ are locally Lipschitz. The control input
$u$ is constrained by the hyperbox
\begin{equation}
    u \in \Ubox := \left\{ u \in \mathbb{R}^{p}\mid \underline{u}_{j}\le u_{j}
    \le \overline{u}_{j},\; j = 1, \ldots, p \right\},
\end{equation}
where $\underline{u}_{j}, \overline{u}_{j}\in \mathbb{R}$ denote the lower and upper bounds on each control component, respectively, which is consistent with the physical constraints of many systems in practice, such as motor torque limits. We denote the solution of system \eqref{eq:sys} with initial condition $x(0) = x_{0}$ and control input $u$ as $\phi(t, x_{0}, u)$, which is defined on the maximal interval of existence $I \subset [0, \infty)$. In this work, we assume that the maximal interval of existence is $[0, \infty)$ for all $x_{0}\in \X$ and $u \in \Ubox$.

\subsection{State constraints and smooth approximation}
\label{sec:state_constraints} Consider multiple state constraints defined by a finite set of continuously differentiable functions $\{h_{i}\}_{i=1}^{N}$, where $h_{i}:\X \to \mathbb{R}$, $i = 1, \ldots , N$. The safe set $\C_{i}$ associated with $h_{i}$ is defined as $\C_{i}:= \{x \in \X \mid h_{i}(x) \le 1\}$, and the unsafe set $\D_{i}$ is defined as $\D_{i}:= \{x \in \X \mid h_{i}(x) > 1\}$. Apparently, the overall safe region $\C_{\mathrm{safe}}$ is the intersection of all $\C_{i}$, i.e., $\C_{\mathrm{safe}}= \bigcap_{i=1}^{N}\C_{i}$, which can also be written as the point-wise maximum of all $h_{i}$ as follows:
\begin{equation}
\C_{\mathrm{safe}}= \{x \in \X \mid h_{max}\le 1\},
\end{equation}
where $h_{max}:= \max_{i=1, \ldots, N}h_{i}(x)$ is not a continuously differentiable function in general. The system is said to be safe with respect to the constraints $h_{i}$ if for any initial state $x_{0}\in \C_{\mathrm{safe}}$, there exists a control input $u \in \Ubox$ such that the solution $\phi (t, x_{0}, u)$ remains in $\C_{\mathrm{safe}}$ for all $t \ge 0$. Let the interior of $\C_{\mathrm{safe}}$ be denoted as $\mathrm{Int}(\C_{\mathrm{safe}})$, which is assumed to be nonempty. 
Let $\xe$ be the unique equilibrium point of \eqref{eq:sys} in $\C_{\mathrm{safe}}$, and suppose that $\xe\in \mathrm{Int}(\C)$.
We then define the safe region of attraction (ROA) under a given state-feedback
control policy $u = \kappa(x) \in \Ubox$ as
\begin{align}
    \mathcal{D}_{\mathrm{safe}}:= \Big\{ & x_{0}\in \C_{\mathrm{safe}} \mid \; \phi(t, x_{0}, \kappa) \in \C_{\mathrm{safe}}, \; \forall t \ge 0, \nonumber \\
     & \text{and}\quad \lim_{t\to\infty}\phi(t, x_{0}, \kappa) = \xe\Big\}.
\end{align}

In this work, we adopt the softmax relaxation idea of the safe set $\C_{\mathrm{safe}}$ from \cite{liu2025computing}, which approximates the pointwise maximum function $h_{max}$ with a smooth function as follows:
\begin{equation}
    h_{\mathrm{soft}}(x; \tau) := \frac{1}{\tau}\log \left( \sum_{i=1}^{N}e^{\tau h_{i}(x)}
    \right),
\end{equation}
where $\tau > 0$ is a tunable parameter. It can be shown that
\[
    \C:=\{h_{\mathrm{soft}}\le 1\} \;\subseteq\; \{h_{\max}\le 1\} \;=\; \C_{\mathrm{safe}}
    .
\]
Denote the boundary of $\C$ by
$\partial \C := \{x \in \X \mid h_{\mathrm{soft}}(x) = 1\}$.

\subsection{Linear temporal logic}

Let $AP$ be a finite set of atomic propositions. We consider the stutter-invariant fragment of LTL, denoted by $\LTL$, which excludes the next operator $\bigcirc$~\cite{clarke1999model}.
We further assume that all formulas are in positive normal form~\cite{clarke1999model}, meaning that negation is only allowed directly in front of atomic propositions.
The syntax of $\LTL$ over $AP$ is defined inductively as
\[
    \varphi ::= \top \mid \bot \mid \pi \mid \neg \pi \mid \varphi_{1}\wedge \varphi
    _{2}\mid \varphi_{1}\vee \varphi_{2}\mid \varphi_{1}\,\mathbf{U}\, \varphi
    _{2}\mid \varphi_{1}\,\mathbf{R}\, \varphi_{2},
\]
where $\pi \in AP$. Here, $\top$ and $\bot$ denote \emph{true} and \emph{false}, respectively; $\neg$, $\wedge$, and $\vee$ denote \emph{negation}, \emph{conjunction}, and \emph{disjunction}; and $\mathbf{U}$ and $\mathbf{R}$ denote the temporal operators \emph{until} and \emph{release}, respectively. The derived temporal operators $\mathbf{F}$ (\emph{eventually}) and $\mathbf{G}$  (\emph{always}) are defined by
$\mathbf{F}\varphi := \top \,\mathbf{U}\, \varphi$ and
$\mathbf{G}\varphi := \bot \,\mathbf{R}\, \varphi$.

Let $\sigma = \sigma_{0}\sigma_{1}\sigma_{2}\cdots$ be an infinite word over the alphabet $2^{AP}$, where each $\sigma_{i}\subseteq AP$. The satisfaction relation $\sigma,i \models \varphi$, meaning that $\sigma$ satisfies $\varphi$ at position $i$, is defined inductively as follows \cite{li2019robustly}:
\begin{itemize}
    \item $\sigma, i \models \top \Leftrightarrow \text{True}$;

    \item $\sigma, i \models \bot \Leftrightarrow \text{False}$;

    \item $\sigma, i \models \pi \Leftrightarrow \pi \in \sigma_{i}$;

    \item $\sigma, i \models \neg \pi \Leftrightarrow \pi \notin \sigma_{i}$;

    \item $\sigma, i \models \varphi_{1}\wedge \varphi_{2}\Leftrightarrow \sigma
        , i \models \varphi_{1}\;\land\; \sigma, i \models \varphi_{2}$;

    \item $\sigma, i \models \varphi_{1}\vee \varphi_{2}\Leftrightarrow \sigma
        , i \models \varphi_{1}\;\lor\; \sigma, i \models \varphi_{2}$;

    \item $\sigma, i \models \varphi_{1}\mathbf{U}\varphi_{2}\Leftrightarrow$
        \[
            \exists j \ge i, \;\left( \sigma, j \models \varphi_{2}\;\land\; \forall
            k \in [i,j),\, \sigma, k \models \varphi_{1}\right);
        \]

    \item $\sigma, i \models \varphi_{1}\mathbf{R}\varphi_{2}\Leftrightarrow$
        \[
            \forall j \ge i,\, \left( \sigma, j \models \varphi_{2}\;\lor\; \exists
            k \in [i,j),\, \sigma, k \models \varphi_{1}\right).
        \]
\end{itemize}
We write $\sigma \models \varphi$ if and only if $\sigma,0 \models \varphi$.


  A word $\sigma$ is said to be in \emph{lasso form} if it can be written as $\sigma
    =\sigma_{\mathrm{pre}}(\sigma_{\mathrm{suf}})^{\omega},$ where $\sigma_{\mathrm{pre}}$
    and $\sigma_{\mathrm{suf}}$ are finite words over $2^{AP}$. Here, $\sigma_{\mathrm{pre}}$
    is called the \emph{prefix}, $\sigma_{\mathrm{suf}}$ is called the \emph{suffix},
    and $(\sigma_{\mathrm{suf}})^{\omega}$ denotes the infinite repetition of $\sigma
    _{\mathrm{suf}}$. We use $\len(\cdot)$ to denote the length of a finite word.

Let $L:\X \to 2^{AP}$ be a labeling function, where $L(x)\subseteq AP$
denotes the set of atomic propositions satisfied at state $x\in \X$. Following
\cite{wongpiromsarn2015automata}, for each atomic proposition $\pi \in AP$, define
$\lbl{\pi}:=\{x \in \X \mid \pi \in L(x)\}.$ For any subset of atomic propositions
$P \in 2^{AP}$, define $\lbl{P}:=\{x \in \X \mid L(x)=P\}$. 
Hence, $\lbl{P}$ is the set of states that satisfy all and only the propositions
in $P$. Note that for any two distinct labels $P,Q \in 2^{AP}$, $\lbl{P}\cap \lbl{Q}= \emptyset$.

\begin{defn}
    [Trace of a continuous-time trajectory \cite{wongpiromsarn2015automata}]
    Given an initial condition $x_{0}\in \X$ and a control input $u \in \Ubox$,
    let $\phi(\cdot, x_{0}, u): \mathbb{R}_{\ge 0}\to \X$ be the corresponding trajectory of the continuous-time system~\eqref{eq:sys}.
    An infinite word $\sigma=\sigma_{0}\sigma_{1}\sigma_{2}\cdots$ over
    $2^{AP}$ is called a trace of $\phi(\cdot,x_{0},u)$ with respect to $\LTL$ if there exists a sequence of time instants $t_{0}t_{1}t_{2}\cdots$ such that
    \begin{enumerate}
        \item $t_{0}=0$, $t_{m}<t_{m+1}$ for all $m\in\mathbb{N}_0$, and $t_{m}
            \to\infty$ as $m\to\infty$;

        \item $\phi(t_{m},x_{0},u)\in \lbl{\sigma_m}$ for all $m \in \mathbb{N}_0$;

        \item if $\sigma_{m}\neq \sigma_{m+1}$ for each $m \in \mathbb{N}_0$,
            then there exists $t^{\prime}\in [t_{m},t_{m+1}]$ such that
            $\phi(t,x_{0},u)\in \lbl{\sigma_m}$ for all
            $t\in [t_{m},t^{\prime})$, $\phi(t,x_{0},u)\in \lbl{\sigma_{m+1}}$
            for all $t\in (t^{\prime},t_{m+1}]$, and
            $\phi\left(t^{\prime}, x_{0}, u\right) \in \lbl{\sigma_m}\cup \lbl
            {\sigma_{m+1}}$.
    \end{enumerate}
\end{defn}

\begin{defn}[\cite{wongpiromsarn2015automata}]
\label{def:traj_satisfaction}
 Given an initial condition $x_{0}\in \X$
and an input $u \in \Ubox$, the trajectory $\phi(\cdot,x_{0},u)$ satisfies an $\LTL$ formula $\varphi$, denoted by $\phi(\cdot,x_{0},u)\models
\varphi$, if the trace $\sigma$ of $\phi(\cdot,x_{0},u)$ satisfies $\varphi$, namely, $\sigma\models\varphi$.
\end{defn}

\yt{
\begin{defn}[Winning set for an LTL specification] 
\label{def:winning_set_ltl}
Given the control-affine system~\eqref{eq:sys}
and an $\LTL$ formula $\varphi$, let $\mathcal{M}$ be a finite memory set. A finite-memory control strategy is denoted by $\kappa:\X\times\mathcal{M}\to\Ubox$. Then, the winning set of system~\eqref{eq:sys} with respect to
$\varphi$ is defined as
\begin{equation}
\begin{aligned}
    \Win(\varphi) := & \Bigl\{x_{0}\in \X \;\Big|\; \exists \mathcal{M}, \exists \;\kappa:\X\times\mathcal{M}\to\Ubox              \\
& \qquad \qquad \text{ such that }\phi(\cdot,x_{0},\kappa)\models\varphi \Bigr\}.
\end{aligned}
\end{equation}
\end{defn}
That is, $\Win(\varphi)$ is the set of all initial states for which there exists a finite-memory control strategy such that the corresponding closed-loop trajectory satisfies $\varphi$.

The finite memory set $\mathcal M$ will be specified in Section~\ref{subsec:problemformulation} to encode the progress of the LTL specification during the closed-loop execution. 
}

\subsection{Problem formulation}
\label{subsec:problemformulation}

Consider the continuous-time system~\eqref{eq:sys} and an $\LTL$ specification $\varphi$ defined over a set of atomic propositions $AP$. Since computing the exact winning set $\operatorname{Win}(\varphi)$ is generally intractable, the objective of this work is to compute a certified inner approximation $\widehat{\operatorname{Win}}(\varphi) \subseteq \operatorname{Win}(\varphi)$,
and to synthesize a switching state-feedback control policy such that, for every initial condition $x(0) \in \widehat{\operatorname{Win}}(\varphi)$, the resulting closed-loop trajectory satisfies $\varphi$.


In this work, given a lasso-form word $\sigma \models \varphi$ and the labeling function $L$ associated with $AP$, the task of satisfying  $\varphi$ can be decomposed into a finite sequence of safe-stabilization subtasks induced by $\sigma$. The overall synthesis problem is then reduced to computing CLBF-based inner approximations of the local winning sets for these subtasks and patching the corresponding local controllers.

The computed set $\widehat{\Win}(\varphi)$ and the associated switching policy together provide a correct-by-construction motion planner: for any initial condition in this certified set, the planner generates a closed-loop trajectory satisfying the prescribed $\LTL$ specification.





\section{LTL specification decomposition}
\label{sec:decomposition}   
Consider an $\LTL$ formula $\varphi$ over $AP$. Let
$
\sigma=\sigma_{\mathrm{pre}}(\sigma_{\mathrm{suf}})^{\omega}
$ be a word over $2^{AP}$ such that $\sigma \models \varphi$. In this work, we interpret $\sigma$ as a reference trace for system~\eqref{eq:sys}. Since $\sigma \models \varphi$, synthesizing a feedback control policy whose closed-loop trajectory induces the trace $\sigma$ ensures the satisfaction of the specification $\varphi$. Following the sequential structure of $\sigma$, we decompose the synthesis problem into safe-stabilization subtasks, each associated with a pair of consecutive letters of $\sigma$.


Specifically, for each consecutive pair $(\sigma_m, \sigma_{m+1})$ in the lasso-form word $\sigma$, with $m\in\mathbb{N}_0$, let $\Xm \subseteq \X$ denote the local state-space region associated with the transition from $\sigma_m$ to $\sigma_{m+1}$. We formulate such a transition as a local reach-avoid control problem and define its transition tuple as
\begin{equation}
    \Tm := (\mathcal{S}_m,\mathcal{G}_m,\mathcal{A}_m),
    \label{eq:transition_set}
\end{equation}
where $\mathcal{S}_m, \mathcal{G}_m, \mathcal{A}_m  \subset \Xm$ denote the source set, the goal set, and the forbidden set, respectively.

Since $\sigma$ is in lasso form, the induced infinite transition sequence
$\{\Tm\}_{m\in\mathbb{N}_0}$ admits a finite prefix-suffix representation. Accordingly, we let $M_{\mathrm{pre}}=\len(\sigma_{\mathrm{pre}})$ denote the number of prefix transitions, including the transition from the last prefix symbol to the first suffix symbol when the suffix is nonempty, and let $M_{\mathrm{suf}}=\len(\sigma_{\mathrm{suf}})$ denote the number of transitions in one suffix cycle, including the wrap-around transition from the last suffix symbol to the first one. We then write the corresponding finite pattern as $\{\Tm\}_{m=0}^{M-1}$, where $M=M_{\mathrm{pre}}+M_{\mathrm{suf}}$.

Therefore, realizing the reference word $\sigma$ is reduced to a sequence of transitions $\{\Tm\}_{m=0}^{M-1}$, each requiring the system to move from $\Sm$ to $\Gm$ while avoiding $\Am$. The following assumptions are imposed.

    \begin{assumption}
        \label{assum:transition_wellposed} For each transition $\Tm = (\Sm,\Gm,\Am
        )$, we assume that
        \begin{itemize}
            \item $\Sm \subseteq \lbl{\sigma_m}$ is nonempty and satisfies $\Sm \cap \Am = \emptyset$;


            \item $\Gm \subseteq \lbl{\sigma_{m+1}}$ is nonempty, connected and closed, and satisfies $\Gm \cap \Am = \emptyset$;

            \item $\Am$ is closed.
        \end{itemize}
    \end{assumption}

    \begin{assumption}
        \label{assum:goal_set_metric} 
        For each transition $\Tm =(\Sm, \Gm, \Am)$, there exist a designated center point $x_{m}^{\star}\in \operatorname{int}(\mathcal{G}_{m})$ and a continuous metric $d_{m}:\Xm\times\Xm\to\mathbb{R}_{\ge 0}$ such that
        \begin{equation}
            \label{eq:goal_set_metric}\mathcal{G}_{m}=\{x \in \Xm \mid d_{m}(x,x_{m}
            ^{\star})\le \delta_{m}\} ,
        \end{equation}
        where $\delta_{m}>0$ is a scalar constant defining the boundary of the goal set. 
    \end{assumption}

    \begin{defn}[Winning set for a word-induced transition] Consider a transition $\Tm=(\Sm,\Gm,\Am)$ induced by the pair of consecutive symbols $(\sigma_{m},\sigma_{m+1})$ in the reference word $\sigma$, as defined in~\eqref{eq:transition_set}. The winning set of $\Tm$ is defined as 
    \begin{equation}
    \begin{aligned}
    \Win_{\sigma}(\Tm) :=&\Bigl\{ x_{0}\in \Sm \;\Big|\; \exists\, \kappa_{m}:\Xm \to\Ubox,\ \exists\, \tau_{m}\ge 0 \\&\text{such that }\phi(t,x_{0},\kappa_{m})\notin \Am,\ \forall t\in[0,\tau_{m}], \\&\text{and}\quad \phi(\tau_{m},x_{0},\kappa_{m})\in \Gm \Bigr\}.
    \end{aligned}
    \end{equation} 
    Here, $\tau_{m}$ denotes the duration of the transition from $\sigma_{m}$ to $\sigma_{m+1}$, $\kappa_{m}:\Xm\to\Ubox$ is a local feedback control strategy, and $\phi(\cdot,x_{0},\kappa_{m})$ is the corresponding trajectory of system~\eqref{eq:sys}.
    \end{defn} 

\yt{
For the finite transition pattern $\{\Tm\}_{m=0}^{M-1}$ induced by the reference word $\sigma$, we instantiate the finite memory set in Definition~\ref{def:winning_set_ltl} as $\mathcal M := \{0,\ldots,M-1\}$ where each element $m\in\mathcal M$ indicates the currently active local transition task $\Tm$. Given local feedback control strategies $\{\kappa_m\}_{m=0}^{M-1}$, the corresponding global finite-memory control strategy is defined by $
\kappa(x,m)=\kappa_m(x)$, $m\in\mathcal M$. Let $\zeta:\mathbb N_0\to\mathcal M$ denote the lasso index map, with $\zeta(0)=0$, which follows the prefix transitions and then loops over the suffix transitions. For each $n\in\mathbb N_0$, when the current memory element is $m=\zeta(n)$ and the trajectory reaches $\mathcal{G}_{\zeta(n)}$, then the memory element is updated to $\zeta(n+1)$.

\begin{proposition}\label{prop:local_to_global_ltl} Consider an $\LTL$ formula $\varphi$, let $\sigma=\sigma_{\mathrm{pre}}(\sigma_{\mathrm{suf}})^{\omega}$ be a reference word such that $\sigma\models\varphi$. Let $\{\Tm\}_{m=0}^{M-1}$, where $M=M_{\mathrm{pre}}+M_{\mathrm{suf}}$, be the corresponding sequence of transitions induced by $\sigma$, with winning sets $\Win_{\sigma}(\Tm)$. Let $\zeta:\mathbb{N}_0\to\mathcal M$ be the lasso index map defined above. Suppose that $x_{0}\in \Win_{\sigma}(\mathcal{T}_{0})$ and \begin{equation}
\mathcal{G}_{\zeta(n)} \subseteq \Win_{\sigma}(\mathcal{T}_{\zeta(n+1)}), \qquad n \in \mathbb{N}_0.
\label{eq.Gn_sub_Win}
\end{equation} 
Then there exists a finite-memory control strategy $\kappa:\X\times\mathcal M\to\Ubox$ for the system~\eqref{eq:sys}, obtained by sequentially patching the local strategies  $\{\kappa_{m}\}_{m=0}^{M-1}$, such that the corresponding trajectory satisfies $\phi(\cdot,x_{0},\kappa)\models\varphi.$
\end{proposition}

\begin{proof}
We construct the closed-loop trajectory recursively.Since $\zeta(0)=0$ and $x_{0}\in \Win_{\sigma}(\mathcal{T}_{0})\subseteq \mathcal{S}_{0}\subseteq \lbl{\sigma_{0}}$, the initial symbol $\sigma_{0}$ is satisfied. By the definition of $\Win_{\sigma}(\mathcal{T}_{0})$ and~\eqref{eq.Gn_sub_Win}, there exists a local strategy $\kappa_{0}$ and a finite time $\tau_{0}\ge 0$ such that $\phi(t,x_{0},\kappa_{0})\notin \mathcal{A}_{0}$, $\forall t\in[0,\tau_{0}],$ and $x_{1}:=\phi(\tau_{0},x_{0},\kappa_{0})\in \mathcal{G}_{0}\subseteq \Win_{\sigma}(\mathcal{T}_{1})\subseteq \lbl{\sigma_1}.$

Proceeding inductively, suppose that for some $n \geq 1$,
$x_{n}\in \Win_{\sigma}(\mathcal{T}_{\zeta(n)})
\subseteq \mathcal{S}_{\zeta(n)}
\subseteq \lbl{\sigma_{\zeta(n)}}.$
Then, by the definition of $\Win_{\sigma}(\mathcal{T}_{\zeta(n)})$ and~\eqref{eq.Gn_sub_Win}, there exists a local strategy $\kappa_{\zeta(n)}$ and a finite time $\tau_{n}\ge 0$ such that
$\phi(t,x_{n},\kappa_{\zeta(n)})\notin \mathcal{A}_{\zeta(n)}$ for all
$t\in[0,\tau_{n}]$, and $x_{n+1}:=\phi(\tau_{n},x_{n},\kappa_{\zeta(n)})
\in \mathcal{G}_{\zeta(n)}
\subseteq \Win_{\sigma}(\mathcal{T}_{\zeta(n+1)})
\subseteq \lbl{\sigma_{\zeta(n+1)}}.
$

Hence, the local strategies $\{\kappa_{m}\}_{m=0}^{M-1}$ can be patched sequentially to generate a trajectory whose induced trace is $\sigma$. Since $\sigma \models \varphi$, the induced trace of $\phi(\cdot,x_0,\kappa)$ satisfies $\varphi$. By Definition~\ref{def:traj_satisfaction}, we conclude that $\phi(\cdot,x_0,\kappa)\models\varphi$.
\end{proof} 
}

Consequently, each transition $\Tm$ can be formulated as a bounded-input safe stabilization subproblem: steer the state toward the designated center $x_{m}^{\star}$ while avoiding $\Am$. The transition is declared finished when the state enters the goal set $\Gm$, whose reaching criterion is given by \eqref{eq:goal_set_metric}.
In the following section, we present the patching CLBF method to address this issue, adopted from \cite{liu2025formal}.

\section{Safe Stabilization with Patching CLBF}
\label{sec:patch_clbf} 
For a given reach-avoid transition $\Tm$, we first present the definitions of CLFs and CBFs with bounded control inputs, which are the building blocks for the patching CLBFs. We assume that $\xs$ is the unique equilibrium point in $\Xm$, i.e., $f(\xs) = 0$, in the absence of
control input ($u=0$). $\Am = \X_{m}/ \C_{\mathrm{safe}}^{m}$, where $\C_{\mathrm{safe}}^{m}$ is
the safe set defined in Section \ref{sec:state_constraints} with $\xs\in \mathrm{Int}(\C^m_{\mathrm{safe}})$ for $\Tm$ with a finite set of $h_{i}^{m}$, and we denote its smooth
approximation by $\C^{m}$.

\begin{defn}
    A continuously differentiable function $V:\Xm \to \mathbb{R}$ is a CLF for
    system \eqref{eq:sys} such that $V(\xs) = 0$, and there exists a control
    input $u \in \Ubox$, $\forall x \in \Xm \setminus \{\xs\}$,
    \begin{align}
        V(x) > 0, \quad \text{and}                  \\
        L_{f}V(x) + L_{g}V(x) u < 0, \label{eq:clf}
    \end{align}
    where $L_{f}V(x) := \nabla V(x) \cdot f(x)$ and $L_{g}V(x) := \nabla V(x)
    \cdot g(x)$ are the Lie derivatives of $V$ along $f$ and $g$,
    respectively.
\end{defn}
\begin{defn}
    A continuously differentiable function $h:\Xm \to \mathbb{R}$ is a CBF for
    system \eqref{eq:sys} with respect to $\Cm$ if there exists a control
    input $u \in \Ubox$, $\forall x \in \Cm$,
    \begin{align}
        h(x) \le 1, \quad \text{and}                      \\
        \quad L_{f}h(x) + L_{g}h(x) u < 0. \label{eq:cbf}
    \end{align}
\end{defn}
Following \cite{liu2025computing,liu2025formal}, we have the following proposition
about the strict compatible condition between the CLF $V$ and the CBF $h$.
\begin{proposition}[\cite{liu2025formal}]
\label{prop:compatible} 
Let $V:\Xm \to \mathbb{R}$ be a candidate CLF and
$h:\Xm \to \mathbb{R}$ be a candidate CBF for system \eqref{eq:sys} with respect to $\Cm$. If the following conditions hold:
\begin{enumerate}
\item For every $x \in \Cm\setminus \{\xs\}$, there exists
$u\in\Ubox$ such that
\begin{equation}
    \label{eq:CLFV}L_{f}V(x) + L_{g}V(x) u < 0.
\end{equation}
Note that this condition is equivalent to
\begin{equation}
    \label{eq:clf_box}
    \begin{aligned}
         & (h(x)\le 1 \wedge x\neq \xe) \;\Longrightarrow\;              \\
         & L_{f}V(x) + \sum_{j=1}^{p}\!\left(m_{j}\,(L_{g}V(x))_{j}- r_{j}\,\big|(L_{g}V(x))_{j}\big| \right) < 0,
    \end{aligned}
\end{equation}
where $m_{j}=\tfrac{\underline{u}_j+\overline{u}_j}{2}$ and
$r_{j}=\tfrac{\overline{u}_j-\underline{u}_j}{2}$.

\item For every $x \in \partial \Cm$, there exists $u\in\Ubox$ such
that
\begin{equation}
    \label{eq:CLFBoundary}L_{f}V(x) + L_{g}V(x) u < 0,
\end{equation}
and
\begin{equation}
    \label{eq:CBFBoundary}L_{f}h(x)+L_{g}h(x) u < 0.
\end{equation}
Given the hyperbox constraint on the control input, the above two
conditions \eqref{eq:CLFBoundary} and \eqref{eq:CBFBoundary} are
equivalent to
\begin{equation}
\label{eq:farkas-box}
\begin{aligned}
 & \Bigl(h(x)=1 \\
 & \ \wedge\ \left\{ \begin{aligned}&\lambda_{1},\lambda_{2},\ \lambda_{j}^{\pm}\ge 0 \quad (j=1,\dots,p),\\
 &\lambda_{1}+\lambda_{2}+\sum_{j=1}^{p}(\lambda_{j}^{+}+ \lambda_{j}^{-}) =1\end{aligned}\right. \\
 & \ \wedge\ \lambda_{1}L_{g}V(x)+\lambda_{2}L_{g}h(x) +\sum_{j=1}^{p}(\lambda_{j}^{+}- \lambda_{j}^{-})\,e_{j}= 0\Bigr) \\
 & \Longrightarrow\; \lambda_{1}L_{f}V(x)+\lambda_{2}L_{f}h(x) < \sum_{j=1}^{p}(\lambda_{j}^{+}\overline u_{j}-\lambda_{j}^{-}\underline{u}_{j}).
\end{aligned}
\end{equation}
Here $e_{j}\in \mathbb{R}^{1 \times p}$ denotes the $j$-th standard
basis row vector.
\end{enumerate}
Then, $V$ and $h$ are said to be strictly compatible.
\end{proposition}

We direct the readers to \cite{liu2025formal} for the detailed proof of
Proposition \ref{prop:compatible}. It is worth noting that conditions \eqref{eq:clf_box}
for the CLF and \eqref{eq:farkas-box} for the compatibility are first-order formulas,
which can be verified by off-the-shelf SMT solvers, such as dReal~\cite{gao2013dreal},
which provides the formal guarantees. Suppose that $\Cm$ is compact, we have
the following theorem.
\begin{thm}
    [\cite{liu2025computing, liu2025formal}] \label{thm:compatibility} Let $V$
    be a CLF and $h$ be a CBF for system \eqref{eq:sys} with respect to $\Cm$.
    If $V$ and $h$ are strictly compatible as defined in Proposition \ref{prop:compatible},
    then there exists a continuously differentiable function
    $W:\,\Xm \ra\Real$, which is a CLBF for system \eqref{eq:sys}, satisfying:
    \begin{enumerate}
        \item $W$ is positive definite, and $W(\xs)=0$.
        \item $\Cm =\set{x\in\Xm \mid W(x)\le 1}$.

        \item For every $x \in \Cm \setminus \{\xs\}$, there exists
            $u\in\Ubox$ such that
            \begin{equation}
                \label{eq:CLFW}L_{f}W(x) + L_{g}W(x) u < 0.
            \end{equation}
    \end{enumerate}
\end{thm}
\rz{We refer the reader to~\cite{liu2025formal} for the detailed proof of this theorem.}
To construct such a CLBF $W$, for some $\varepsilon > 0$, we define the smooth
bump function
\[
    b=
    \begin{cases}
        \exp\!\Bigl(-\dfrac{1}{\varepsilon^2 - (h(x)-1)^2}+ \dfrac{1}{\varepsilon^2}\Bigr), & 1-\varepsilon < h(x) < 1, \\[4pt]
        1,                               & h(x) \ge 1, \\
        0,                               & h(x) \le 1-\varepsilon.
    \end{cases}
\]
As discussed in \cite{liu2025formal}, the compatibility also holds on the inner boundary of $\Cm$, i.e., $\partial\mathcal{C}_{\varepsilon}^{-}:= \{x \in \Xm \mid 1-\varepsilon \le h ( x) \le 1\}.$ By rescaling $V$ if necessary, we can assume that $V(x) \le 1$, $\forall x \in \Cm$, and $V(x) \le h(x)$ $\forall x \in \partial\mathcal{C}_{\varepsilon}^{-}$. Then, $W$ can be constructed as follows \cite{liu2025computing}:
\begin{equation}
    \label{eq:CLBF}W(x) := (1 - b(x)) V(x) + b(x) h(x),
\end{equation}
and the 1-sublevel set of $W$, denoted as $\Wm$, captures \rz{an estimate of the} safe ROA of $\xs$. Thus, we have the approximated winning set $\hwint = \Wm \setminus \Gm$, and $\hwint \subseteq \Win_{\sigma}(\Tm)$ . \rz{By guaranteeing 
\begin{equation} \label{eq:subset}
    \mathcal{G}_{m-1} \subset \hwint \subset \Win_{\sigma}(\Tm),
\end{equation}
 we have that \eqref{eq.Gn_sub_Win} still holds with the approximations of local winning sets.}
\begin{rem}
It is worth noting that one can also use Lyapunov stability theory for set stability with respect to $\Gm$ to construct the CLBF $W$ for the transition $\Tm$ given some other assumptions, which turns the local specification into a reachability problem while keeping safe. In this case, we do not need the assumption for $\Gm$ in Assumption \ref{assum:goal_set_metric}. However, due to the asymptotic stability nature of the CLF-based method, which might result in an extremely slow convergence close to $\Gm$.
\end{rem}

\section{LTL Control Synthesis with Patching CLBFs}
\label{sec:LTL_with_CLBFs}
In this section, we consider an $\LTL$ specification $\varphi$ satisfying the conditions in Proposition \ref{prop:local_to_global_ltl} with a lasso-form trace $\sigma =\sigma_{\mathrm{pre}} (\sigma_{\mathrm{suf}})^{\omega}$. We then present the process of synthesizing the patching CLBF-based controller for the LTL specification $\varphi$ by sequentially switching the local CLBF-based controllers for the reach-avoid transitions $\{\Tm\}^{M-1}_{m=0}$.

For each transition $\Tm$, let $W$ be the patching CLBF in \eqref{eq:CLBF}.
We synthesize the bounded controller proposed in \cite{leyva2013global} for hyperrectangle input constraints, in the form $u^{\eta}(x)=\rho^{\eta}(x)\odot \omega(x)$, where $\omega$ is the vertex control and $\rho^{\eta}$ is a component-wise rescaling, and $\odot$ denotes component-wise multiplication.

For system \eqref{eq:sys} with $u \in \R^{p}$, let
\begin{equation}
    a_{W}(x)=L_{f}W(x)=\nabla W(x)\cdot f(x),
\end{equation}
and for each $j\in\{1,\ldots,p\}$,
\begin{equation}
    b_{W,j}(x)=-L_{g_j}W(x)=-\nabla W(x)\cdot g_{j}(x),
\end{equation}
where $\nabla W(x)= b\nabla h +(1-b)\nabla V + (h-V)\nabla b.$

Under box constraints $u_{j}\in[\underline{u}_{j}, \overline{u}_{j}]$, we set
\begin{equation}
    r_{j}(x)=
    \begin{cases}
        \overline{u}_{j},   & b_{W,j}(x)\ge 0, \\
        -\underline{u}_{j}, & b_{W,j}(x)<0,
    \end{cases}
\end{equation}
\begin{equation}
    \omega_{j}(x)=r_{j}(x)\,\mathrm{sign}(b_{W,j}(x)),
\end{equation}
\begin{equation}
    \beta(x)=\sum_{j=1}^{p}r_{j}(x)\lvert b_{W,j}(x)\rvert.
\end{equation}
Next, we define
\begin{equation}
    \lambda(x)=1-\frac{\lvert a_{W}(x)\rvert+a_{W}(x)}{2\beta(x)},
\end{equation}
\begin{equation}
    \phi_{j}(x)=\frac{r_{j}(x)\lvert b_{W,j}(x)\rvert}{\beta(x)},
\end{equation}
\begin{equation}
    \tau_{j}^{\eta}(x) =
    \begin{cases}
        p\dfrac{\ln(\lambda(x))}{\lambda(x)}-\eta r_{j}(x)\lvert b_{W,j}(x)\rvert, & \beta(x)>0, \\
        0,                                                                                & \beta(x)=0,
    \end{cases}
\end{equation}
where $\eta > 0$ is a positive constant.
Then the continuous rescaling function $\rho_{j}^{\eta}(x)$ is synthesized as:
\begin{equation}
    \rho_{j}^{\eta}(x) =
    \begin{cases}
        1 - (1 - \Gamma(x) \phi_{j}(x)) e^{\tau_{j}^{\eta}(x) \phi_{j}(x)}, & \text{if }\Pi_j(x) > 0; \\
        0,                          & \text{if }\Pi_j(x) = 0.
    \end{cases}
\end{equation}
where $\Gamma(x) = \frac{|a_{W}(x)| + a_{W}(x)}{2\beta(x)}$ and $\Pi_j(x) = r_{j}
(x) |b_{W,j}(x)|$. Consequently, the CLBF-based controller for the transition
$\Tm$ is given by
\begin{equation}
    \label{eq:clbf_control}
    u^{\eta}(x)=\left[u_{1}^{\eta}(x),\ldots
    ,u_{p}^{\eta}(x)\right]^{\top}
\end{equation}
with $u_{j}^{\eta}(x)=\rho_{j}^{\eta}(x )\omega_{j}(x)$ for each
$j\in\{1,\ldots,p\}$.

With the CLBF-based controllers for  transitions $\{\Tm\}^{M-1}_{m=0}$, we can sequentially patch them together to generate a switching control strategy for the LTL specification $\varphi$. \yt{For each transition
$\Tm$, let $\kappa_m(x) := u^{\eta,m}(x)$ where $x\in\Xm$ and $u^{\eta,m}$ is the CLBF-based controller synthesized by~\eqref{eq:clbf_control}. The global finite-memory control strategy is then defined as $\kappa(x,m)=u^{\eta,m}(x)$, $m\in\mathcal{M}$.} The synthesis process is:
\begin{itemize}
    \item \textbf{Step 1:} Decompose the LTL specification $\varphi$ with  a lasso-form word  $ \sigma=\sigma_{\mathrm{pre}}(\sigma_{\mathrm{suf}})^{\omega}$ into a sequence of transitions $\{\Tm\}_{m=0}^{M-1}$ as in Section \ref{sec:decomposition}.

    \item \textbf{Step 2:} For the first transition, compute the safe ROA estimate with a patching CLBF in Section \ref{sec:patch_clbf} and its controller $u^{\eta, 0}$ with \eqref{eq:clbf_control}. The corresponding 1-level set $\Wo_1$ provides an inner approximation of the winning set of $\varphi$, $\hwin(\varphi)$.

    \item \textbf{Step 3:} Repeat Step 2 for each transition $\Tm$ to construct a patching CLBF and obtain the corresponding approximated winning set $\Wm$ and the associated local controller \yt{$\kappa_m(x)=u^{\eta,m}(x)$, using \eqref{eq:clbf_control}.}

    \item \textbf{Step 4:}
    \rz{By Proposition~\ref{prop:local_to_global_ltl} and \eqref{eq:subset},} the global control synthesis is achieved by aggregating the local winning-set approximations $\Wm$ via set union. By establishing a switching strategy that applies the appropriate local controller within each respective region, we can synthesize correct-by-construction motion planning for all the points in $\hwin(\varphi)$. 
\end{itemize}
We summarize the proposed framework in Algorithm \ref{alg:patched_clbf_ltl}.

\begin{algorithm}[t]
\caption{Patching CLBF-Based Control under $\LTL$ Specification $\varphi$}
\label{alg:patched_clbf_ltl}
\KwIn{System~\eqref{eq:sys} and word $\sigma=\sigma_{\mathrm{pre}}(\sigma_{\mathrm{suf}})^\omega \models \varphi$}
\KwOut{Closed-loop trajectory $\phi(\cdot; x_0, u)$ and control law $u$}

$M_{\mathrm{pre}} \gets 0$ if $\sigma_{\mathrm{pre}}$ is empty, else $\len(\sigma_{\mathrm{pre}})$\;
$M_{\mathrm{suf}} \gets 0$ if $\sigma_{\mathrm{suf}}$ is empty, else $\len(\sigma_{\mathrm{suf}})$ \;
$M \gets M_{\mathrm{pre}} + M_{\mathrm{suf}}$\; 
\yt{
\tcp{Define the memory-index map $\zeta$;}
\SetKwFunction{MemoryIndex}{MemoryIndex}
\SetKwProg{Fn}{Function}{:}{}
\Fn{\MemoryIndex{$n$}}{
\uIf{$0 \leq n < M_{\mathrm{pre}}$}{
    \Return{$n$}\;
}
\Else{
    \Return{$M_{\mathrm{pre}} + ((n-M_{\mathrm{pre}})\bmod M_{\mathrm{suf}})$}\;
}
}
}
\BlankLine
\tcp{Offline synthesis phase:}
\yt{Construct transitions $\{\mathcal{T}_m\}_{m=0}^{M-1}$ that satisfy Assumption~\ref{assum:transition_wellposed}} \;
\For{$m=0$ \KwTo $M-1$}{
Construct the safe ROA estimate $\Wm$ for $\mathcal{T}_m$ using~\eqref{eq:CLBF} \;
\yt{\If{$m\geq 1$}{Verify that $\mathcal{G}_{m-1}$ satisfies~\eqref{eq:subset}\;}}
Store $\Wm$ and controller $u^{\eta,m}(\cdot)$\;
}

\BlankLine
\tcp{Online control phase:}
\yt{
$n \gets 0$\tcp*{stage counter;}
}
$x(0)\gets x_0$\;
\While{execution continues}{
$m \gets \MemoryIndex(n)$ \tcp*[l]{$m=\zeta(n)$;}
Apply $u(t)=u^{\eta, m}(x(t))$\;
Evolve $\dot{x}(t)=f(x(t))+g(x(t))u(t)$\;
\If{$x(t)\in \mathcal{G}_{m}$}{
    \yt{$n \gets n+1$ \;
    }
}
}
\end{algorithm}

\yt{
\begin{rem}
We would like to emphasize that the proposed motion planning framework can accommodate trajectory-tracking deviations, including those caused by moderate exogenous noise or disturbances, through rapid online replanning. When the current state deviates from the planned trajectory during tracking, the current state can be used as a new initial condition, and the control synthesis can be re-executed to generate a new valid path. The formal satisfaction guarantee is restricted to the nominal model and holds provided that the new initial state remains inside the certified local winning set $\hwint$ associated with the current local specification $\Tm$. Under this condition, the replanned trajectory is guaranteed to satisfy $\varphi$.
\end{rem}
}

\section{Case Studies}

In this section, we evaluate the efficacy of the proposed approach on two simulations and one real experiment with the Crazyflie quadrotor. The patching CLBF step is done with the LyZNet~\cite{liu2024lyznet} toolbox on a server node equipped with an Intel(R) Xeon(R) Gold 6326 CPU @ 2.90 GHz (32 cores). Comparisons were conducted on a server with $2\times$ Intel Xeon Silver 4114 CPUs and 128 GB RAM, with each run allocated 2 CPU threads.


Firstly, we consider the attitude tracking problem for a rigid spacecraft. The system dynamics are governed by Euler's rotational equations~\cite{krstic1999inverse}:
    \begin{equation}
        J\dot{x}= -x \times Jx + u \notag
    \end{equation}
    where $J \in \mathbb{R}^{3 \times 3}$ is the positive-definite, symmetric inertia matrix. For an asymmetric rigid body, we assume $J = \text{diag}(10, 15, 20)$. $x = [x_{1}, x_{2}, x_{3}]^{\top}$ denotes the angular velocities of the spacecraft in the body-fixed frame, and $u = [u_{1}, u_{2}, u_{3}]^{\top}$ denotes the input control torques. We consider the reach-avoid LTL specification as follows:
    \begin{equation} \label{eq:ltl_spec1}
    \varphi = \mathbf{F}(\mathrm{Green}) \land \mathbf{G}(\neg \mathrm{Red_1}) \land \mathbf{G}(\neg \mathrm{Red_2}),
    \end{equation}
which requires the system to eventually reach the green region and always avoid the two red regions. 

As a matter of fact, we can have a quadratic CLF locally for this system in the form of $V(x) = x^{\top} P x$ by linearizing the system at the equilibrium point \rz{and solving the continuous-time algebraic Riccati equation with $Q$ and $R$ as identity matrices of appropriate dimensions. This identical CLF configuration is maintained throughout the subsequent numerical examples}. For such systems with known CLFs, we only need to define clearly the barrier functions with respect to the obstacles, and then the patching CLBF can be constructed easily by the method in Section \ref{sec:patch_clbf}. As shown in Fig.~\ref{fig:spacecraft}, we have $h_1 = 1- ((x_{1} - 1)^2 + (x_{2} - 0.6)^2 + x_{3}^2 - 0.5^2)$ for $\mathrm{Red_1}$ (upper right) and $h_2 = 1 - ((x_{1} + 1)^2 + (x_{2} + 1)^2 + x_{3}^2 - 0.5^2)$ for $\mathrm{Red_2}$ (lower left), respectively. The goal set $\mathrm{Green}$ is defined as a ball region centered at the origin with a radius of 0.1. Then, we construct a smooth barrier function with $\tau = 5.2$, and consequently patch it with the CLF with $\varepsilon = 0.5$.
Sampling within the winning set, all 20 trajectories initialized from different initial conditions successfully realize the LTL specification with the computed controller.

    \begin{figure}[!h]
        \centering
    \includegraphics[width=0.95\linewidth]{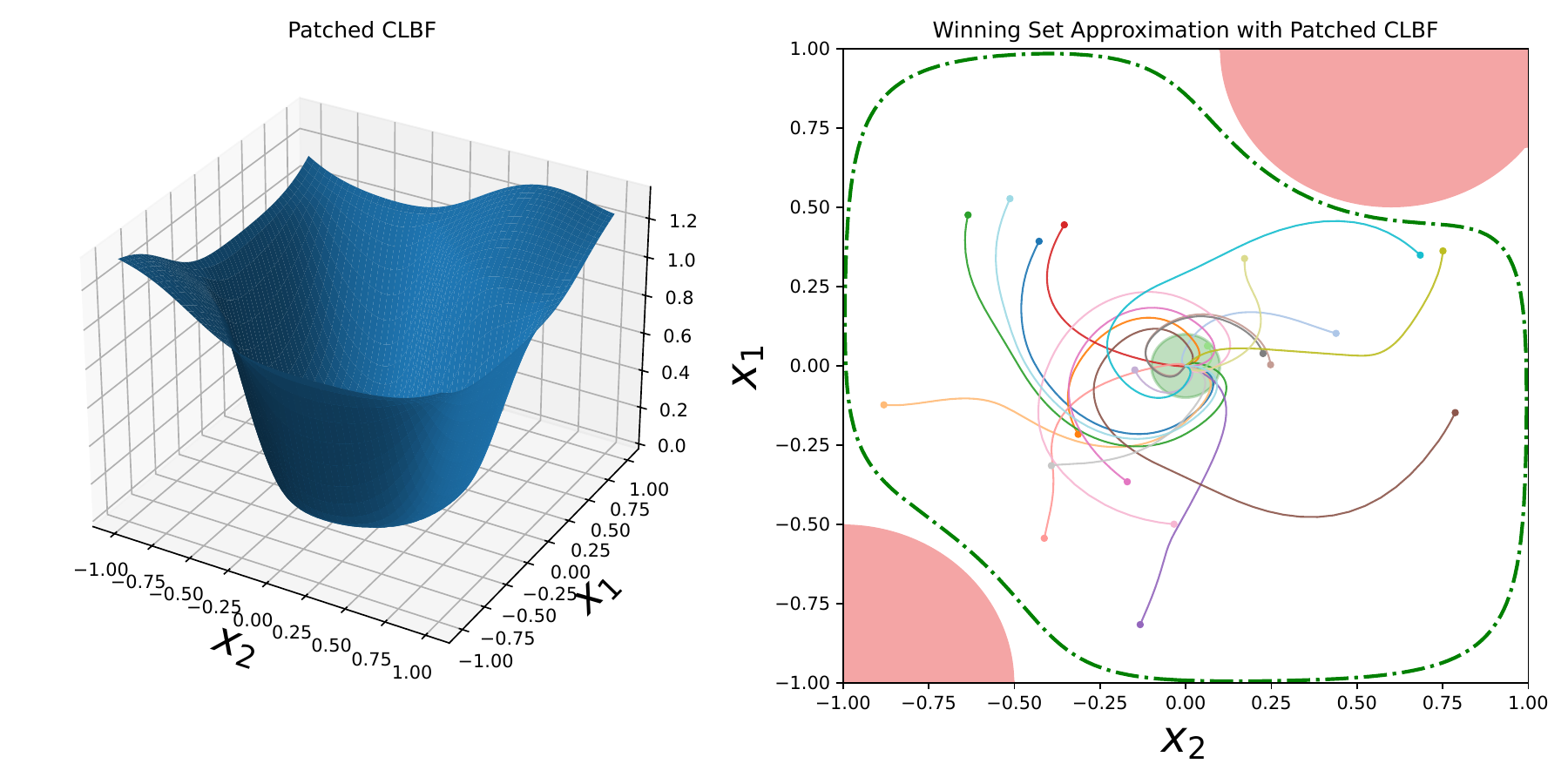
        }
        \caption{The certified CLBF for the LTL specification~\eqref{eq:ltl_spec1}. The dashed green curve approximates the winning set under bounded inputs $u \in [-3, 3]^3$, within which all the trajectories satisfying $\varphi$.}
        \label{fig:spacecraft}
    \end{figure}


Secondly, we consider an omnidirectional robot~\cite{das2024prescribed} with state 
$x=[x_1,x_2,x_3]^\top$ and control input 
$u=[v_1,v_2,\omega]^\top \in [-3, 3]^3$. The dynamics are given by
$$
    \begin{bmatrix}
        \dot{x}_{1} \\
        \dot{x}_{2} \\
        \dot{x}_{3}
    \end{bmatrix}
    =
    \begin{bmatrix}
        \cos x_{3} & -\sin x_{3} & 0 \\
        \sin x_{3} & \cos x_{3}  & 0 \\
        0          & 0           & 1
    \end{bmatrix}
    \begin{bmatrix}
        v_{1}  \\
        v_{2}  \\
        \omega
    \end{bmatrix}.
$$
As shown in Fig.~\ref{fig:omni_ltl_scene}, the robot needs to complete a cyclic transfer task in a gated workspace, where $\mathrm{gate}_3$ and $\mathrm{gate}_4$ open automatically upon robot proximity and are thus omitted from the LTL specification. Starting from the processing station $r_4$, the robot alternates between two routes: one through $r_5$, $r_2$, and $\mathrm{gate}_1$ to reach $r_1$, and the other through $r_3$, $r_7$, and $\mathrm{gate}_2$ to reach $r_6$. In each route, the robot picks up a key before passing through the corresponding gate and then returns to $r_4$. The outbound and return routes are denoted by $\mathfrak{y}_1,\mathfrak{y}_2$ and $\mathfrak{y}_1',\mathfrak{y}_2'$, respectively.
$$
\begin{aligned}
\mathfrak{y}_1 &= r_4 \land \mathbf{F}\big(r_5 \land \mathbf{F}(r_2 \land \mathbf{F}(\mathrm{gate}_1 \land \mathbf{F}r_1))\big),\\
\mathfrak{y}^{\prime}_1 &= r_1 \land \mathbf{F}\big(\mathrm{gate}_1 \land \mathbf{F}(r_2 \land \mathbf{F}(r_5 \land \mathbf{F}r_4))\big),\\
\mathfrak{y}_2 &= r_4 \land \mathbf{F}\big(r_3 \land \mathbf{F}(r_7 \land \mathbf{F}(\mathrm{gate}_2 \land \mathbf{F}r_6))\big),\\
\mathfrak{y}^{\prime}_2 &= r_6 \land \mathbf{F}\big(\mathrm{gate}_2 \land \mathbf{F}(r_7 \land \mathbf{F}(r_3 \land \mathbf{F}r_4))\big).
\end{aligned}
$$

The overall LTL specification requires the two transfer cycles to be repeatedly executed in an alternating manner:
$$
\varphi
= \mathbf{G}\left(\neg \mathrm{obs}\right) 
\land
\mathbf{G}\mathbf{F}
\left(
\mathfrak{y}_1
\land
\mathbf{F}\left(
\mathfrak{y}^{\prime}_1
\land
\mathbf{F}\left(
\mathfrak{y}_2
\land
\mathbf{F}\;\mathfrak{y}^{\prime}_2
\right)
\right)
\right),
$$
where $\mathrm{obs}=\bigvee^9_{i=1} \mathrm{obs}_i$. In this task, we need to decompose the overall LTL specification into $16$ transitions as described in Sec.~\ref{sec:decomposition}. Regarding the barrier functions, we use again a \emph{log-sum-exp} smooth approximation, and then construct the barrier functions for each $\Tm$ with $\tau = 4$ as well as $0.5$ safety margin for the obstacles to further ensure the robot's safety. As illustrated in Fig.~\ref{fig:omni_ltl_scene}, for the first subtask $\varphi_1=r_4 \wedge \mathrm{F}(r_5)$, we can compute the approximated winning set $\hwin(\varphi)$ as the curve in the purple dashed line. As an example, the orange dashed line is the approximated winning set $\widehat{\Win_{\sigma}}(\mathcal{T}_{10})$ for the subtask $\varphi_{10}=r_3 \wedge \mathrm{F}(r_7)$, which captures almost the whole winning set of this local specification. Consequently, for each subtask, we can immediately synthesize controllers for any states within the approximated winning set, which provides robustness and flexibility without solving QPs on the fly.

\begin{figure}[!h]
    \centering    
    \includegraphics[width=0.9\linewidth]{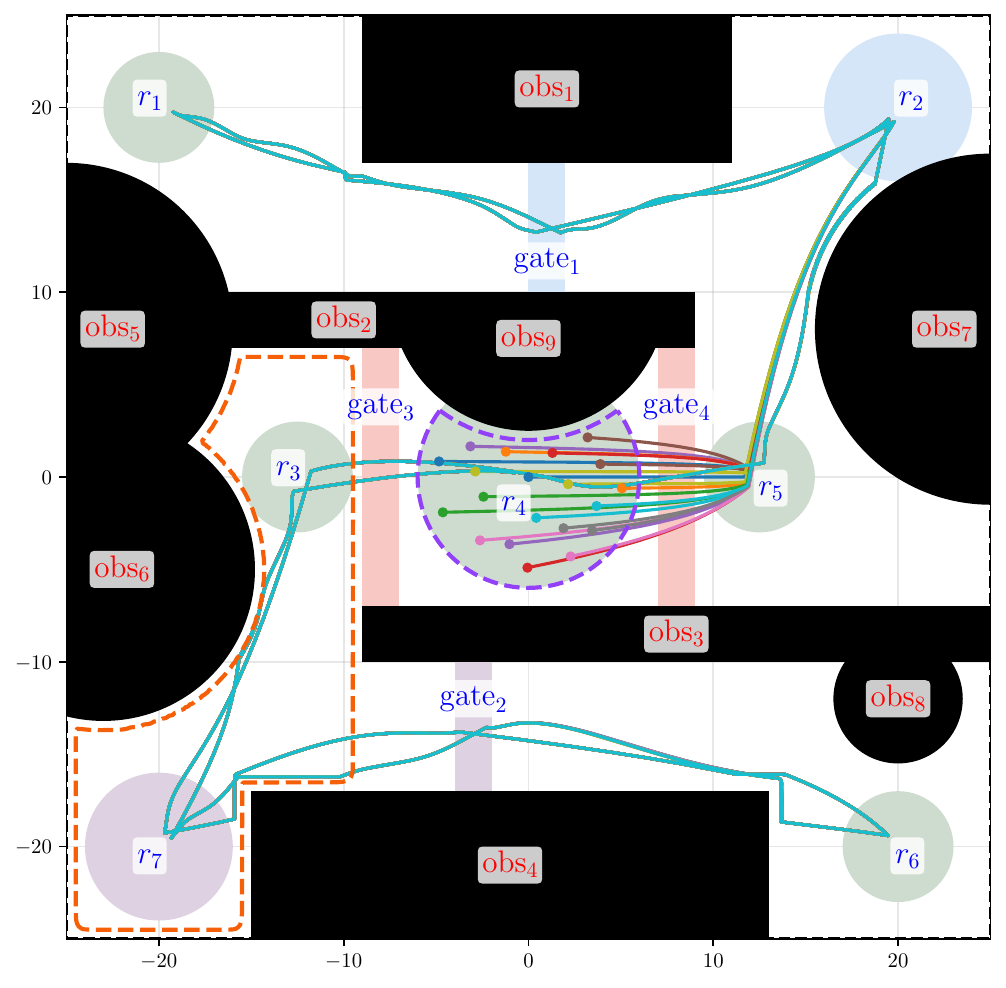}
     
     \caption{Twenty trajectories of the omnidirectional robot initialized from different states in a complex planar environment under the LTL specification using the proposed method, with $\hwin(\varphi)$ in purple dashed line. The controller is executed at $50~\mathrm{Hz}$.}
    \label{fig:omni_ltl_scene}
\end{figure}

\begin{table}[!h]
\centering
\caption{Comparison of trajectory performance and computational cost.
Values are reported as mean $\pm$ sample standard deviation over 20 trajectories.}
\renewcommand{\arraystretch}{2.3}
\setlength{\tabcolsep}{1pt}
\begin{tabular}{|c|c|c|c|}
\hline
Method & \makecell{CPU \\ runtime $\bar{T}(\mathrm{s})$} & \makecell{CPU runtime\\per step $\bar{T}_{\mathrm{step}} (\mathrm{ms})$} & \makecell{Control-energy cost\\per step $\bar{E}_{\mathrm{step}}$} \\
\hline
\hline
Patched CLBF & $\mathbf{23.984 \pm 0.129}$ & $\mathbf{0.632 \pm 0.002}$ & $\mathbf{0.366 \pm 0.007}$  \\
\hline
QP-CLF-CBF & $29.795 \pm 0.360$ & $2.60 \pm 0.01$  & $2.196 \pm 0.013$  \\
\hline
\end{tabular}
\label{tab:comparison}
\end{table}

With 20 different initial conditions sampled in the approximated winning set, we further compare the proposed controller against the standard CBF-QP baseline with the CLF-based controller via Sontag's formula. For each trajectory $i\in\{1,\ldots,20\}$, let $T_i$, $M_i$, and $u_i(k)$ denote the CPU runtime, the number of control steps, and the control input at the $k$-th step, respectively. Since $M_i$ varies with the initial condition, we average the per-trajectory metrics so that each trajectory contributes equally. Specifically, we compute the CPU runtime as $\bar{T}=\frac{1}{20}\sum_{i=1}^{20}T_i$, the CPU runtime per control step as $\bar{T}_{\mathrm{step}}=\frac{1}{20}\sum_{i=1}^{20}\frac{T_i}{M_i}$, and the control-energy cost per control step as $\bar{E}_{\mathrm{step}}=\frac{1}{20}\sum_{i=1}^{20}\frac{1}{M_i}\sum_{k=1}^{M_i}\|u_i(k)\|_2^2$. As shown in Table~\ref{tab:comparison}, the proposed controller achieves higher computational efficiency by avoiding QP solving at every control step, thereby reducing the overall computation time. The proposed controller also yields a lower control-energy cost.

Lastly, we test our method on a Crazyflie quadrotor for a planar task with the following LTL specification: $\varphi = a_0 \land \mathbf{F}a_1 \land \mathbf{G}\mathbf{F}a_2 \land \mathbf{G}\mathbf{F}a_3 \land \mathbf{G}(\neg \mathrm{obs})$ where $\mathrm{obs}=\vee_{i=1}^3 \mathrm{obs}_i$. We use the single integrator with $u \in [-0.3, 0.3]^2$ for planning with the proposed method, while the tracking is done with its default controller. This specification requires the trajectory to start from region $a_0$, avoid the three obstacle regions, eventually visit region $a_1$, and repeatedly visit regions $a_2$ and $a_3$. The simulation scenario in Webots\footnote{\url{https://www.cyberbotics.com/}} and resulting trajectory are shown in Fig.~\ref{fig:crazyflie_scene}, and the corresponding real experiment scenario and trajectory are shown in Fig.~\ref{fig:crazyflie_real_scene}. The results demonstrate the effectiveness of the proposed method in generating trajectories that satisfy complex temporal logic specifications while respecting control constraints. \yt{Although the simplified planning model used in the real experiment may introduce deviations, the recorded real-world trajectory follows the planned trajectory and completes the specified task sequence, demonstrating the practical trackability of the generated trajectory.}



\begin{figure}[!ht]
\centering    
\includegraphics[width=0.8\linewidth]{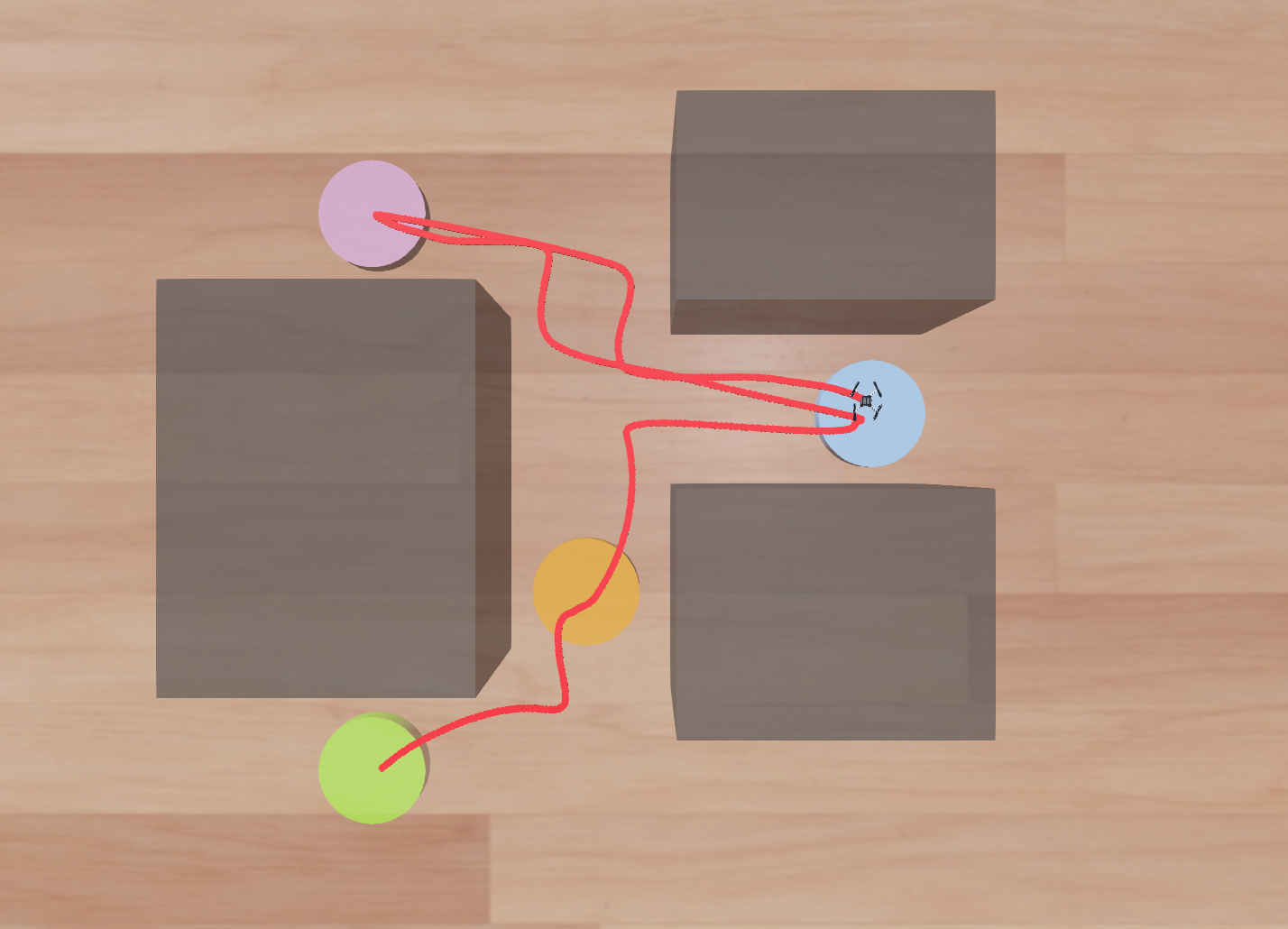}
\caption{Webots simulation scenario and executed trajectory for the Crazyflie quadrotor. The black, green, orange, blue, and pink regions correspond to obstacles, $a_0$, $a_1$, $a_2$, and $a_3$, respectively. 
The red curve denotes the executed trajectory.}
\label{fig:crazyflie_scene}
\end{figure}

\begin{figure}[!ht]
    \centering    \includegraphics[width=0.8\linewidth, trim=0 1.5cm 0 0, clip]{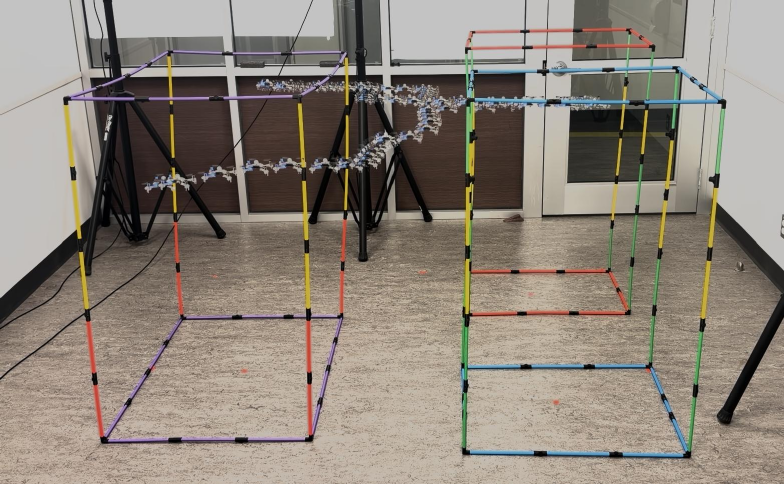}
     \caption{The real experiment and executed trajectory for the Crazyflie quadrotor, where the boxes are obstacles.}
    \label{fig:crazyflie_real_scene}
\end{figure}


\section{Conclusion and Discussions}
This work presents a patched-CLBF framework for synthesizing switching feedback controllers for continuous-state systems subject to $\LTL$ specifications and bounded inputs. By linking feasible temporal-logic trajectories to CLBF level sets, the method decomposes the synthesis problem into a sequence of safe stabilization tasks and enables efficient online execution. Simulations and a Crazyflie quadrotor experiment demonstrate the effectiveness of the approach. Future work will focus on reducing the conservativeness of the certificate construction and improving scalability to higher-dimensional systems, where certificate search and patch verification become more challenging.

Moreover, the proposed framework is sound but not complete: successful synthesis guarantees specification satisfaction, while failure does not preclude the existence of satisfying controllers. This incompleteness stems from the conservative, sufficient-condition-based structure of certificate verification. While recent work derives necessary-and-sufficient conditions for specific finite-trace formulations, these do not extend to general continuous-state $\LTL$ synthesis \cite{niu2023necessary}. By contrast, the proposed framework provides the formal guarantee that whenever the procedure succeeds, the synthesized closed-loop behavior satisfies the given specification \cite{wongpiromsarn2015automata}.
\bibliographystyle{unsrt}
\bibliography{CDC26}

\end{document}